\documentclass[letterpaper,twocolumn,10pt]{article}
\usepackage{usenix2019_v3}

\usepackage[utf8]{inputenc}
\usepackage{booktabs} 
\usepackage{graphicx}
\usepackage{subfigure}
\usepackage{units}
\usepackage{blindtext}
\usepackage{xcolor,xspace}
\usepackage{amsmath}
\usepackage{amssymb}
\usepackage{paralist}
\usepackage{amsmath}
\usepackage{tikz}
\usepackage{booktabs}
\usepackage{pifont}
\usepackage{multirow}
\usepackage{hhline}
\usepackage{balance}
\usepackage{footnote}
\usepackage{logicproof}
\usepackage{authblk}
\usepackage{amsthm}
\usetikzlibrary{calc}
\usetikzlibrary{decorations.pathreplacing}

\usepackage{arydshln}
\usepackage{mdframed}
\usepackage{pgf-umlsd}
\usetikzlibrary{patterns}

\usepackage{expl3}
\expandafter\def\csname ver@l3regex.sty\endcsname{} 

\usepackage[T1]{fontenc}
\usepackage{
beramono,
listings,
textcomp
}

\usetikzlibrary{arrows,automata}
\usetikzlibrary{calc}
\usetikzlibrary{arrows,calc,shapes,decorations.pathreplacing}

\usepackage{pgfplots}
\usepackage{pgfplotstable}
\usepackage{caption}
\pgfplotsset{compat=newest}

\usepackage{listings}
\usepackage[ruled]{algorithm2e}
\usepackage{algorithmic}
\usepackage{authblk}

\LinesNumbered

\usepackage[
n,
advantage,
operators,
sets,
adversary,
landau,
probability,
notions,
logic,
ff,
mm,
primitives,
events,
complexity,
asymptotics,
keys]{cryptocode}

\newcommand{\acro}{{{\sf\it GAROTA}}\xspace}
\newcommand{\acrolong}{{{\underline{G}eneralized \underline{A}ctive \underline{R}oot-\underline{O}f-\underline{T}rust} \underline{A}rchitecture}\xspace}

\newcommand{\dev}{{\ensuremath{\sf{\mathcal Prv}}}\xspace}

\newcommand{\func}{{\ensuremath{\sf{\mathcal F}}}\xspace}
\newcommand{\trigger}{{\ensuremath{\sf{trigger}}}\xspace}

\renewcommand\adv{\ensuremath{\sf{\mathcal Adv}}\xspace}

\mathchardef\mhyphen="2D

\newcommand{\dmaaddr}{\ensuremath{DMA_{addr}}\xspace}
\newcommand{\dmaen}{\ensuremath{DMA_{en}}\xspace}

\newcommand{\modMem}{\ensuremath{\mathsf{modMem}}}
\renewcommand{\mod}{\ensuremath{\mathsf{mod}}}
\newcommand{\disable}{\ensuremath{\mathsf{disable}}}

\newcommand{\ignore}[1]{}

\newtheorem{definition}{Definition}
\newtheorem{theorem}{Theorem}

\author[]{Esmerald Aliaj}
\author[]{Ivan De Oliveira Nunes}
\author[]{Gene Tsudik}
\affil[]{University of California, Irvine}

\definecolor{listinggray}{gray}{0.9}
\definecolor{lbcolor}{rgb}{0.9,0.9,0.9}
\definecolor{Darkgreen}{rgb}{0,0.4,0}
\begin{document}

\title{GAROTA: \acrolong}

\maketitle

\begin{abstract}
Embedded (aka smart or IoT) devices are increasingly popular and becoming ubiquitous. 
Unsurprisingly, they are also attractive attack targets for exploits and malware. Low-end embedded 
devices, designed with strict cost, size, and energy limitations, are especially challenging 
to secure, given their lack of resources to implement sophisticated security services, 
available on higher-end computing devices. To this end, several tiny Roots-of-Trust (RoTs) 
were proposed to enable services, such as remote verification of device's software state 
and run-time integrity. Such RoTs operate reactively: they can prove whether a desired action 
(e.g., software update, or execution of a program) was performed on a specific device. 
However, they can not guarantee that a desired action will be performed, 
since malware controlling the device can trivially block access to the RoT by ignoring/discarding 
received commands and other trigger events. This is a major and important problem 
because it allows malware to effectively ``brick'' or incapacitate a potentially huge 
number of (possibly mission-critical) devices.

Though recent work made progress in terms of incorporating more active behavior atop 
existing RoTs, it relies on extensive hardware support --Trusted Execution Environments 
(TEEs) which are too costly for low-end devices. In this paper, we set out to systematically 
design a minimal active RoT for tiny low-end MCU-s. We begin with the following questions: 
(1) What functions and hardware support are required to guarantee actions in the 
presence of malware?, (2) How to implement this efficiently?, and (3) What security benefits 
stem from such an active RoT architecture? We then design, implement, formally verify, 
and evaluate \acro~: \acrolong. We believe that \acro is the 
first clean-slate design of an active RoT for low-end MCU-s. We show how \acro
guarantees that even a fully software-compromised low-end MCU performs a desired action. 
We demonstrate its practicality by implementing \acro in the context of three types of 
applications where actions are triggered by: sensing hardware, network events and timers. 
We also formally specify and verify \acro functionality and properties.
\end{abstract}

\section{Introduction}\label{sec:intro}
The importance of embedded systems is hard to overestimate and their use in critical settings is projected to rise sharply~\cite{projection}. 
Such systems are increasingly inter-dependent and used in many settings, including household, office, factory, automotive, 
health and safety, as well as national defense and space exploration. Embedded devices are usually deployed in large quantities
and for specific purposes. Due to cost, size and energy constraints, they typically cannot host complex security mechanisms. 
Thus, they are an easy and natural target for attackers that want to quickly and efficiently cause harm on an organizational, 
regional, national or even global, level. 
Fundamental trade-offs between security and other priorities, such as cost or performance are a recurring theme in the domain 
of embedded devices. Resolving these trade-offs, is challenging and very important.

Numerous architectures focused on securing low-end micro-controller units (MCU-s) by designing small and affordable trust
anchors~\cite{abera2016things}. However, most such techniques {\bf operate passively}. They can prove, to a trusted party, 
that certain property (or action) is satisfied (or was performed) by a remote and potentially compromised low-end MCU. 
Examples of such services include remote attestation~\cite{smart,sancus,vrased,simple,tytan,trustlite}, proofs of remote 
software execution~\cite{apex}, control-flow \& data-flow attestation~\cite{litehax,cflat,lofat,atrium,oat,tinycfa}, as well as 
proofs of remote software update, memory erasure, and system-wide reset\cite{pure,verify_and_revive,asokan2018assured}. 
These architectures  are typically designed to provide proofs that are unforgeable, despite potential compromise of the MCU.

Aforementioned approaches are passive in nature. While they can detect integrity violations of remote devices, they cannot 
guarantee that a given security or safety-critical task will be performed. For example, consider a network comprised of 
a large number (of several types of) simple IoT devices, e.g., an industrial control system. 
Upon detecting a large-scale compromise, a trusted remote controller wants to fix the situation by requiring 
all compromised devices to reset, erase, or update themselves in order to expunge malware. Even if each device has an 
uncompromised, yet passive, RoT, malware (which is in full control of the device's software state) can easily
intercept, ignore, or discard any requests for the RoT, thus preventing its functionality from being triggered. Therefore, the only
way to repair these compromised devices requires direct physical access (i.e, reprogramming by a human) to each device. Beyond
the DoS damage caused by the multitude of essentially ``bricked'' devices, physical reprogramming itself is slow and disruptive, i.e., 
a logistical nightmare.

Motivated by the above, recent research~\cite{proactive1,proactive2} yielded trust anchors that exhibit a more active behavior.
Xu et al.~\cite{proactive1} propose the concept of Authenticated Watch-Dog Timers (WDT), which can enforce periodic execution
of a secure component (an RoT task), unless explicit authorization (which can itself include a set of tasks) is received from a trusted controller.
In~\cite{proactive2} this concept is realized with the reliance on an existing passive RoT (ARM TrustZone), as opposed to a dedicated 
co-processor as in the original approach from~\cite{proactive1}. Both techniques are time-based, periodically and actively triggering 
RoT invocation, despite potential compromise of the host device. (We discuss them in more detail in Section~\ref{sec:app_authWDT}.)

In this paper, we take the next step and design a more general \underline{active} RoT, called \acro: \acrolong.
Our goal is an architecture capable of triggering guaranteed execution of trusted and safety-critical tasks based on
arbitrary events captured by hardware peripherals (e.g., timers, GPIO ports, and network interfaces) of an MCU the 
software state of which may be currently compromised. In principle, any hardware event that causes an interruption 
on the unmodified MCU can trigger guaranteed execution of trusted software in \acro (assuming proper configuration).
In that vein, our work can be viewed as a generalization of concepts proposed in~\cite{proactive1,proactive2}, enabling 
arbitrary events (interruption signals, as opposed to the timer-based approach from prior work) to trigger guaranteed 
execution of trusted functionalities. In comparison, prior work has the advantage of relying on pre-existent hardware, thus 
not requiring any hardware changes. On the other hand, our clean-slate approach, based on a minimal hardware design, 
enables new applications and is applicable to lower-end resource-constrained MCU-s.

At a high level, \acro is based on two main notions: ``{\bf Guaranteed Triggering}'' and ``{\bf Re-Triggering on Failure}''.
The term \trigger is used to refer to an event that causes \acro RoT to take over the execution in the MCU. 
The ``guaranteed triggering'' property ensures that a particular event of interest always triggers execution of \acro RoT.
Whereas,``re-triggering on failure'' assures that, if RoT execution is illegally interrupted for any reason (e.g., attempts to 
violate execution's integrity, power faults, or resets), the MCU resets and the RoT is guaranteed to be the first to execute after subsequent re-initialization.
Figure~\ref{fig:overview} illustrates this workflow.

\begin{figure}
\centering
\includegraphics[width=0.9\columnwidth]{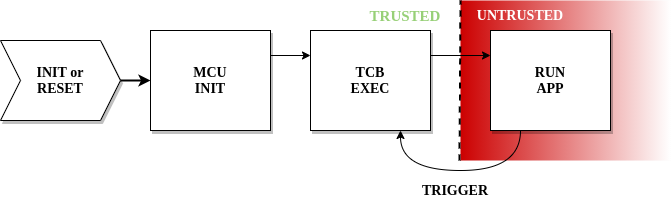}
\vspace{-1mm}
\caption{\acro Software Execution Flow}\label{fig:overview}
\end{figure}

We use \acro to address three realistic and compelling use-cases for the active RoT:
\begin{compactitem}
\item GPIO-TCB: A safety-critical sensor/actuator hybrid, which is guaranteed to sound an alarm 
if the sensed quantity (e.g., temperature, CO$_2$ level, etc) exceeds a certain threshold. 
This use-case exemplifies hardware-based triggering.
\item TimerTCB: A real-time system where a predefined safety-critical task is guaranteed to execute periodically. 
This use-case exemplifies timer-based triggering, which is also attainable by~\cite{proactive1,proactive2}.
\item NetTCB: a trusted component that is always guaranteed to process commands received over the network, 
thus preventing malware in the MCU from intercepting and/or discarding commands destined for the RoT. 
This use-case exemplifies network-based triggering.
\end{compactitem}
In all three cases, the guarantees hold despite potential full compromise of the MCU software state, 
as long as the RoT itself is trusted.

In addition to designing and instantiating \acro with three use-cases, we formally specify \acro goals and requirements
using Linear Temporal Logic (LTL). These formal specifications offer precise definitions for the security offered by \acro 
and its corresponding assumptions expected from the underlying MCU, i.e., its machine model.  This can serve as an 
unambiguous reference for future implementations and for other derived services. Finally, we use formal verification 
to prove that the implementation of \acro hardware modules adheres to a set of sub-properties (also specified in LTL) 
that -- when composed with the MCU machine model -- are sufficient to achieve \acro end-to-end goals.
In doing so, we follow a similar verification approach that has been successfully applied in the context
of \underline{passive} RoT-s~\cite{vrased,apex,rata}.

We implement and evaluate \acro and make its verified implementation (atop the popular low-end MCU TI MSP430)
along with respective computer proofs/formal verification publicly available in~\cite{repo}.

\section{Scope}\label{sec:scope}
This work focuses on low-end embedded MCU-s and on design with minimal hardware requirements.
A minimal design simplifies reasoning about \acro and formally verifying its security properties. 
In terms of practicality and applicability, we believe that an architecture that is cost-effective enough for 
the lowest-end MCU-s can also be adapted (and potentially enriched) for implementations on higher-end devices
with higher hardware budgets, while the other direction is usually more challenging.
Thus, our design is applicable to the smallest and weakest devices based on low-power single-core platform
with only a few KBytes of program and data memory (such as the aforementioned Atmel AVR ATmega and TI MSP430), 
with $8$- and $16$-bit CPUs, typically running at $1$-$16$ MHz clock frequencies, with $\approx64$ KBytes 
of addressable memory. SRAM is used as data memory ranging in size between $4$ and $16$ KBytes, 
while the rest of address space is available for program memory. Such devices usually run software atop 
``bare metal'', execute instructions in place  (physically from program memory), and have no memory management 
unit (MMU) to support virtual memory.

Our initial choice of implementing \acro atop the well-known TI MSP430 low-energy MCU is motivated by availability of a 
well-maintained open-source MSP430 hardware design from OpenCores \cite{openmsp430}. Nevertheless, our design 
and its machine model are applicable to other low-end MCU-s of the same class.

\section{\acro Overview}\label{sec:overview}
\begin{figure}
\centering
\includegraphics[width=0.9\columnwidth]{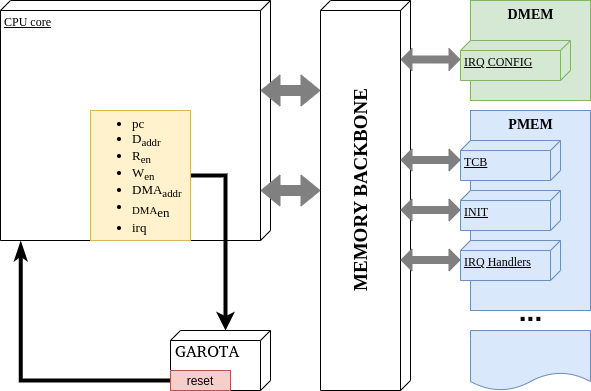}
\vspace{-1mm}
\caption{\acro in the MCU architecture}\label{fig:architecture}
\end{figure}
The goal of \acro is to guarantee eventual execution of a pre-defined functionality \func implemented as a trusted software 
executable. We refer to this executable as \acro trusted computing base (TCB).
\acro is agnostic to the particular functionality implemented by \func, which allows guaranteed execution of 
arbitrary tasks, to be determined based on the application domain; see Section~\ref{sec:applications} for examples.

A \trigger refers to a particular event that can be configured to cause the TCB to execute.
Examples of possible triggers include hardware events from:
\begin{compactitem} \small\it
\item External (usually analog) inputs, e.g., detection of a button press, motion, sound or certain temperature/CO$_2$ threshold.
\item Expiring timers, i.e., a periodic \trigger.
\item Arrival of a packet from the network, e.g., carrying a request to collect sensed data, perform sensing/actuation, or initiate 
a security task, such as code update or remote attestation.
\end{compactitem} \normalsize\rm
If configured correctly, these events cause interruptions, which are used by \acro to guarantee execution of \func.
Since \trigger and TCB implementation are configurable, we assume that these initial configurations are done
securely, at or before initial deployment. \trigger configuration will include the types of interruptions and respective settings 
e.g., which GPIO port, what type of event, its time granularity, etc. At runtime, \acro protects the initial configuration 
from illegal modifications, i.e., ensures correct \trigger behavior. This protection includes preserving interrupt configuration 
registers, interrupt handlers, and interrupt vectors. This way \acro guarantees that \trigger always results in an invocation
of the TCB.

However, guaranteed invocation of the TCB upon occurrence of a \trigger is not sufficient to claim that \func is properly 
performed, since the TCB code (and execution thereof) could itself be tampered with. To this end, \acro provides runtime 
protections that prevent any unprivileged/untrusted program from modifying the TCB code, i.e., the program memory region 
reserved for storing that code. (Recall that instructions execute in place, from program memory). \acro also monitors the 
execution of the TCB code to ensure that:
\begin{compactenum}
\item {\bf Atomicity:} Execution is atomic (i.e., uninterrupted), from the TCB's first instruction (legal entry), to its last instruction (legal exit);
\item {\bf Non-malleability:} During execution, $DMEM$ cannot be modified, other than by the TCB code itself, e.g., no modifications by
other software or DMA controllers.
\end{compactenum}
These two properties ensure that any potential malware residing on the MCU (i.e., compromised software outside TCB or compromised 
DMA controllers) cannot tamper with TCB execution. 

\acro monitors TCB execution and, if a violation of {\em any} property
(not just atomicity and non-malleability)
occurs, it triggers an immediate MCU reset to a default trusted state where TCB code is the first component to execute.
Therefore, any attempt to interfere with the TCB functionality or execution only causes the TCB to recover and re-execute, 
this time with the guarantee that unprivileged/untrusted applications cannot interfere.

Both \trigger configurations and the TCB implementation are updatable at run-time, as long as the updates are performed 
from within the TCB itself. While this feature is not strictly required for security, we believe it provides flexibility/updatability, 
while ensuring that untrusted software is still unable to modify \acro trusted components and configuration thereof.
In Section~\ref{sec:confidentiality}, we also discuss how \acro can enforce TCB confidentiality, 
which is applicable to cases where \func implements cryptographic or privacy sensitive tasks.


Each sub-property in \acro is implemented, and individually optimized, as a separate \acro sub-module.
These sub-modules are then composed and shown secure (when applied to the MCU machine model) using a 
combination of model-checking-based formal verification and an LTL computer-checked proof.
\acro modular design enables verifiability and minimality, resulting in low hardware overhead and significantly 
higher confidence about the security provided by its design and implementation.

As shown in Figure~\ref{fig:architecture}, \acro is implemented as a hardware component that monitors a 
set o CPU signals to detect violations to 
required security properties. As such it does not interfere with the CPU core implementation, e.g., 
by modifying its behavior or instruction set. In subsequent sections we describe these properties in 
more detail and discuss their implementation and verification. Finally, we use a commodity FPGA to 
implement \acro atop the low-end MCU MSP430 and report on its overhead.

\section{\acro in Detail}\label{sec:core}
We now get into the details of \acro. The next section provides some background on LTL
and formal verification; given some familiarity with these notions, it can be skipped without loss of continuity.

\subsection{LTL \& Verification Approach}\label{sec:verif_background}
Formal Verification refers to the computer-aided process of proving that a system (e.g., hardware, software, or protocol) 
adheres to its well-specified goals. Thus, it assures that the system does not exhibit any unintended behavior, 
especially, in corner cases (rarely encountered conditions and/or execution paths) that humans tend to overlook. 

To verify \acro, we use a combination of Model Checking  and Theorem Proving, summarized next.
In Model Checking, designs are specified in a formal computation model (e.g., as Finite State Machines or FSMs) 
and verified to adhere to formal logic specifications. The proof is performed through automated and exhaustive 
enumeration of all possible system states. If the desired specification is found not to hold for specific states 
(or transitions among them), a trace of the model that leads to the erroneous state is provided, and the 
implementation can then be fixed accordingly. As a consequence of exhaustive enumeration, proofs for complex 
systems that involve complex properties often do not scale well due to so-called ``state explosion''.

To cope with that problem, our verification approach (in line with prior work~\cite{vrased,apex}) is to specify 
each sub-property in \acro using Linear Temporal Logic (LTL) and verify each respective sub-module for compliance.
In this process, our verification pipeline automatically converts digital hardware, described at Register Transfer Level 
(RTL) using Verilog, to Symbolic Model Verifier (SMV)~\cite{smv} FSMs 
using Verilog2SMV~\cite{irfan2016verilog2smv}.  The SMV representation is then fed to the well-known 
NuSMV~\cite{nusmv} model-checker for verification against the 
specified LTL sub-properties. Finally, the composition of the LTL sub-properties (verified in the model-checking phase) 
is proven to achieve \acro end-to-end goals using an LTL theorem prover~\cite{spot}. Our verification strategy is 
depicted in Figure~\ref{fig:verif_strategy}.

\begin{figure}
\centering
\includegraphics[height=0.55\columnwidth,width=0.8\columnwidth]{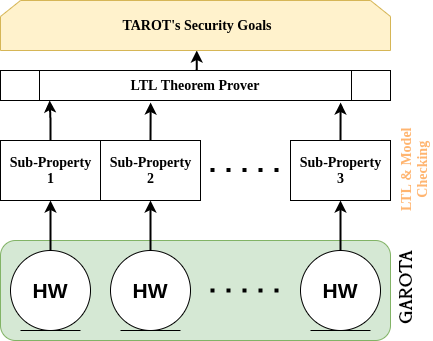}
\vspace{-1mm}
\caption{\acro verification strategy}\label{fig:verif_strategy}
\end{figure}

Regular propositional logic includes propositional connectives, such as: conjunction $\land$, disjunction $\lor$, 
negation $\neg$, and implication $\rightarrow$.  LTL augments it with temporal quantifiers, thus enabling 
sequential reasoning. In this paper, we are interested in the following temporal quantifiers:
\begin{compactitem}
\item \textbf{X}$\phi$ -- ne\underline{X}t $\phi$: holds if $\phi$ is true at the next system state.
\item \textbf{F}$\phi$ -- \underline{F}uture $\phi$: holds if there exists a future state where $\phi$ is true.
\item \textbf{G}$\phi$ -- \underline{G}lobally $\phi$: holds if for all future states $\phi$ is true.
\item $\phi$ \textbf{U} $\psi$ -- $\phi$ \underline{U}ntil $\psi$: holds if there is a future state where $\psi$ holds and
$\phi$ holds for all states prior to that.
\item $\phi$ \textbf{W} $\psi$ -- $\phi$ \underline{W}eak until $\psi$: holds if, assuming a future state where $\psi$ 
holds, $\phi$ holds for all states prior to that. If $\psi$ never becomes true, $\phi$ must hold forever. 
Or, more formally: $\phi \textbf{W} \psi \equiv (\phi \textbf{U} \psi) \lor \textbf{G}(\phi)$
\end{compactitem}
Note that, since \acro TCB is programmable and its code depends on the exact functionality \func for each application domain, 
verification and correctness of any specific TCB code is not within our goals. We assume that the user is responsible for 
assuring correctness of the trusted code to be loaded atop \acro active RoT. 
This assumption is consistent with other programmable (though passive) RoTs, including those targeting higher-end 
devices, such as Intel SGX~\cite{SGX}, and ARM TrustZone\cite{trustzone}. In many cases, we expect the TCB code 
to be minimal (see examples in Section~\ref{sec:applications}), and thus unlikely to have bugs.
\subsection{Notation, Machine Model, \& Assumptions}\label{sec:adv}
\label{sec:MCU_assumptions}
This section discusses our machine and adversarial models. 
We start by overviewing them informally in Sections~\ref{sec:cpu_signals}, \ref{sec:mem_model} and~\ref{sec:trigger_model}). 
Then, Section~\ref{sec:machine_model}, formalizes the machine model using LTL.
For quick-reference, Table~\ref{tab:notation} summarizes the notation used in the rest of the paper.

\begin{table}
\tt \scriptsize
\begin{center}
\begin{tabular}{r p{6.3cm} }
\toprule$PC$&  Current Program Counter value\\
$R_{en}$&  Signal that indicates if the MCU is reading from memory (1-bit)\\
$W_{en}$&  Signal that indicates if the MCU is writing to memory (1-bit)\\
$D_{addr}$&  Address for an MCU memory access (read or write) \\
\dmaen&  Signal that indicates if DMA is currently enabled (1-bit)\\
\dmaaddr&  Memory address being accessed by DMA, if any\\
$gie$&  Global Interrupt Enable: signal that indicates whether or not interrupts are globally enabled (1-bit). \\
$irq$&  Signal that indicates if an interrupt is happening \\
$DMEM$ &  Region corresponding to the entire data memory of the MCU: 
$DMEM =[DMEM_{min}, DMEM_{max}]$.\\
$PMEM$& Region corresponding to the entire program memory of the MCU: 
$PMEM = [PMEM_{min}, PMEM_{max}]$.\\
$TCB$&  Memory region reserved for the TCB's executable implementing 
\func: $TCB = [TCB_{min}, TCB_{max}]$.  $TCB \in PMEM$.   \\
$INIT$&  Memory region containing the MCU's default initialization code. 
\func: $INIT = [INIT_{min}, INIT_{max}]$. $INIT \in PMEM$.   \\
$reset$ & A 1-bit signal that reboots/resets the MCU when set to logical $1$ \\
\bottomrule
\end{tabular}
\end{center}
\vspace{-0.5cm}
\caption{Notation Summary}\label{tab:notation}
\vspace{-1em}
\end{table}

\subsubsection{CPU Hardware Signals}\label{sec:cpu_signals}
\acro neither modifies nor verifies the underlying CPU core/instruction set.
It is assumed that the underlying CPU adheres to its specification and \acro is implemented 
as a standalone hardware module that runs in parallel with the CPU, and enforcing necessary 
guarantees in hardware. The following CPU signals are relevant to \acro:

\noindent {\bf H1 -- \emph{Program Counter (PC):}} $PC$ always contains the address of the 
instruction being executed in the current CPU cycle.

\noindent {\bf H2 -- \emph{Memory Address:}} Whenever memory is read or written by the 
CPU, the data-address signal ($D_{addr}$) contains the address of the corresponding memory 
location. For a read access, a data read-enable bit ($R_{en}$) must be set, while, for a write 
access, a data write-enable bit ($W_{en}$) must be set.

\noindent {{\bf H3 -- \emph{DMA:}} Whenever a DMA controller attempts to access the main 
system memory, a DMA-address signal (\dmaaddr) reflects the address of the memory location 
being accessed and a DMA-enable bit (\dmaen) must be set. DMA can not access memory when 
\dmaen is off (logical zero).

\noindent {\bf H4 -- \emph{MCU Reset:}} At the end of a successful reset routine, all registers 
(including $PC$) are set to zero before resuming normal software execution flow. Resets are 
handled by the MCU in hardware. Thus, the reset handling routine can not be modified.
Once execution re-starts, PC is set to point to the first instruction in the boot section of program 
memory, referred to as $INIT$ (see M2 below). When a reset happens, the corresponding $reset$ 
signal is set. The same signal is also set when the MCU initializes for the first time. An MCU $reset$ 
also resets its DMA controller, and any prior configuration thereof. (DMA) behavior is configured 
by user software at runtime. By default (i.e., after a reset) DMA is inactive.

\noindent {\bf H5 -- \emph{Interrupts:}} Whenever an interrupt occurs, the corresponding $irq$ signal 
is set. Interrupts may be globally enabled or disabled in software. The 1-bit signal $gie$ always 
reflects whether or not they are currently enabled. The default $gie$ state (i.e., at boot or after a reset) 
is disabled (logical zero).

\subsubsection{Memory: Layout \& Initial Configuration}\label{sec:mem_model}
As far as MCU initial memory layout and its initial software configuration (set at, or prior to, its deployment),
the following are relevant to \acro:

\noindent {{\bf M1 -- \emph{PMEM:}} Corresponds to the entire $PMEM$ address space. 
Instructions are executed in place. Hence, at runtime, $PC$ points to the $PMEM$ address 
storing the instruction being executed.

\noindent {\bf M2 -- \emph{INIT:}} Section of $PMEM$ containing the MCU boot segment, i.e., the first 
software to be executed whenever the MCU boots or after a $reset$. We assume $INIT$ code is finite.

\noindent {\bf M3 -- \emph{TCB:}} Section of $PMEM$ reserved for \acro trusted code, i.e., \func. 
TCB is located immediately after $INIT$; it is the first software to execute following successful completion 
of $INIT$.

\noindent {\bf M4 -- \emph{IRQ-Table and Handlers:}} IRQ-Table is located in $PMEM$ and contains 
pointers to the addresses of so-called \emph{interrupt handlers}. When an interrupt occurs, the MCU
hardware causes a jump to the corresponding handler routine. The address of this routine is specified 
by the IRQ-Table fixed index corresponding to that particular interrupt. Handler routines are code 
segments (functions) also stored in $PMEM$.

\noindent {\bf M5 -- \emph{$IRQ_{cfg}$:}} Set of registers in $DMEM$ used to configure specific 
behavior of individual interrupts at runtime, e.g., deadline of a timer-based interrupt, or type of event on a 
hardware-based interrupt.

Note that the initial memory configuration can be changed at run-time (e.g., by malware that infects the 
device,as discussed in Section~\ref{sec:adv_model}), unless it is explicitly protected by \acro 
verified hardware modules.

\subsubsection{Initial Trigger Configuration}\label{sec:trigger_model}
\noindent {\bf T1 -- \trigger:} \acro\ \trigger is configured, at MCU (pre)deployment-time, by setting the 
corresponding entry in IRQ-Table and respective handler to jump to the first instruction in TCB ($TCB_{min}$) 
and by configuring the registers in $IRQ_{cfg}$ with desired interrupt parameters, reflecting the desired 
\trigger behavior; see Section~\ref{sec:applications} for examples. Thus, a \trigger event causes the TCB code to 
execute, as long as the initial configuration is maintained.

Our initial configuration is not much different from a regular interrupt configuration in a typical embedded 
system program. It must correctly point to \acro TCB legal entry point, just as regular interrupts must 
correctly point to their respective handler entry points. For example, to initially configure a timer-based trigger, 
the address in IRQ-Table corresponding to the respective hardware timer is set to point to $TCB_{min}$ 
and the correspondent registers in $IRQ_{cfg}$ are set to define the desired interrupt period.

\subsubsection{Adversarial Model}\label{sec:adv_model}
We consider an adversary \adv that controls \dev's entire software state, including code, and data.
\adv\ can read/write from/to any memory that is not explicitly protected by hardware-enforced access 
control rules. \adv\ might also have full control over all Direct Memory Access (DMA) controllers of \dev. 
Recall that DMA allows a hardware controller to directly access main memory ($PMEM$ or $DMEM$) 
without going through the CPU.

\textbf{Physical Attacks:} physical and hardware-focused attacks are out of scope of \acro. Specifically, we assume that 
\adv\ can not modify induce hardware faults, or interfere with \acro via physical presence attacks and/or side-channels. 
Protection against such attacks is an orthogonal issue, which can be addressed via physical security techniques~\cite{ravi2004tamper}.

\textbf{Network DoS Attacks:} we also consider out-of-scope all network DoS attacks whereby \adv drops traffic to/from \dev, or floods \dev
with traffic, or simply jams communication. Note that this assumption is relevant only to network-triggered events, exemplified by the NetTCB 
instantiation of \acro, described in Section \ref{sec:app_nettcb}. 

\textbf{Correctness of TCB's Executable:} we stress that the purpose of \acro is guaranteed execution of \func, 
as specified by the application developer and loaded onto \acro TCB at deployment time. Similar to existing 
RoTs (e.g., TEE-s in higher-end CPUs) \acro does \textbf{not} check correctness of, 
and absence of implementation bugs in, \func's implementation. In many applications, \func code
is minimal; see examples in Section~\ref{sec:applications}. Moreover, correctness of \func need \textbf{not} be assured. 
Since embedded applications are originally developed on more powerful devices (e.g., general-purpose computers), 
various vulnerability detection methods, e.g., fuzzing~\cite{fuzzer}, static analysis~\cite{costin2014large}, or formal 
verification, can be employed to avoid or detect implementation bugs in \func.
All that can be performed off-line before loading \func onto \acro TCB and the entire issue is orthogonal to \acro functionality.

\subsubsection{Machine Model (Formally)}\label{sec:machine_model}
\begin{figure*}[!ht] 
\begin{mdframed}
\linespread{0.5}
\footnotesize
\begin{definition}\label{def:machine_model}{\underline{Machine Model:}}\\

{\bf Memory {Modifications: } } 

\begin{align}
\begin{split}\label{ltl:mem_write}
{\bf G:}\{\modMem(X) \rightarrow (W_{en} \land D_{addr} \in X) \lor (DMA_{en} \land DMA_{addr} \in X)\}
\end{split}
\end{align}
{\bf Successful Trigger Modification:}

\begin{align}\label{ltl:trigger_mod}
\begin{split}
\hspace*{0.3cm}{\bf G:}\{\mod(\trigger_{cfg}) \rightarrow [(\modMem(PMEM) \lor \modMem(IRQ_{cfg})) \land \neg reset]\}
\end{split}
\end{align}
{\bf Successful Interrupt Disablement:}

\begin{align}
\begin{split}\label{ltl:gie}
\textbf{G:}\{\disable(irq) \rightarrow [\neg reset \land gie \land \neg {\bf X}(gie) \land \neg {\bf X}(reset)]\}
\end{split}
\end{align}
{\bf Trigger/TCB Initialization (\ref{ltl:no_mod_no_prob} \& \ref{ltl:no_mod_no_prob_reset}):}
\begin{align}
\begin{split}\label{ltl:no_mod_no_prob}
\textbf{G:} \{\neg \mod(\trigger_{cfg}) \lor PC \in TCB\} \land \textbf{G:} \{\neg \disable(irq) \lor {\bf X}(PC) \in TCB\} \rightarrow \textbf{G:}\{\trigger \rightarrow F(PC=TCB_{min})\}
\end{split}
\end{align}

\begin{align}
\begin{split}\label{ltl:no_mod_no_prob_reset}
\textbf{G:} \{\neg \modMem(PMEM) \lor PC \in TCB\} \rightarrow \textbf{G:}\{reset \rightarrow F(PC=TCB_{min})\}
\end{split}
\end{align}
\end{definition}
\end{mdframed}
\vspace{-3mm}
\caption{MCU machine model (subset) in LTL.}\label{fig:machine_model}
\end{figure*}

Based on the high-level properties discussed earlier in this section, we now formalize the subset (relevant to \acro)
of the MCU machine model using LTL. Figure~\ref{fig:machine_model} presents our machine model as a 
set of LTL statements.

LTL statement~(\ref{ltl:mem_write}) models the fact that modifications to a given memory address ($X$) can 
be done either via the CPU or DMA. Modifications by the CPU imply setting $W_{en}=1$ and $D_{addr}=X$.
If $X$ is a memory region, rather than a single address, we denoted that a modification happened within 
the particular region by saying that $D_{addr} \in X$, instead. Conversely, DMA modifications to region $X$ require 
$DMA_{en}=1$ and $DMA_{addr} \in X$. This models the MCU behaviors stated informally in {\bf H2} and {\bf H3}.

In accordance with {\bf M4} and {\bf M5}, a successful modification to a pre-configured \trigger implies changing 
interrupt tables, interrupt handlers, or interrupt configuration registers (ICR-s). Since, per {\bf M4}, the first two are located in 
$PMEM$, modifying them means writing to $PMEM$. The ICR is located in a $DMEM$ location denoted  $IRQ_{cfg}$.
Therefore, the LTL statement~(\ref{ltl:trigger_mod}) models a successful misconfiguration of \trigger as requiring 
a memory modification either within $PMEM$ or within $IRQ_{cfg}$, without causing an immediate system-wide 
reset ($\neg reset$). This is because an immediate reset prevents the modification attempt from taking effect (see {\bf H4}).

LTL~(\ref{ltl:gie}) models that attempts to disable interrupts are reflected by $gie$ CPU signal (per {\bf H5}). In order to 
successfully disable interrupts, \adv must be able to switch interrupts from enabled ($gie=1$) to disabled ($\neg {\bf X}(gie)$ -- 
disabled in the following cycle), without causing an MCU $reset$.

Recall that (from {\bf H1}) $PC$ reflects the address of the instruction currently executing. $PC \in TCB$ implies that 
\acro TCB is currently executing. LTL~(\ref{ltl:no_mod_no_prob}) models {\bf T1}. As long as the initial proper configuration 
of \trigger is never modifiable by untrusted software ($\textbf{G:} \{\neg \mod(\trigger_{cfg}) \lor PC \in TCB\}$) and that 
untrusted software can never globally disable interrupts ($\textbf{G:} \{\neg \disable(irq) \lor {\bf X}(PC) \in TCB\}$), 
a \trigger would always cause TCB execution ($\textbf{G:}\{\trigger \rightarrow F(PC=TCB_{min})\}$).
Recall that we assume that the TCB may update -- though not misconfigure -- \trigger behavior, since the TCB is trusted.
Similarly, LTL~\ref{ltl:no_mod_no_prob_reset} states that, as long as $PMEM$ is never modified by untrusted software, a 
$reset$ will always trigger TCB execution (per {\bf H4}, {\bf M2}, and {\bf M3}).

This concludes our formal model of the default behavior of low-end MCU-s considered in this work.

\subsection{\acro End-To-End Goals Formally}
\begin{figure*}[!ht]
\begin{mdframed}
\linespread{0.5}
\footnotesize
\begin{definition}\label{def:trigger} {\underline{Guaranteed Trigger:}}

\begin{align*}
\textbf{G:}\{
\mathsf{trigger} \rightarrow \text{\bf F}(PC=TCB_{min})\}
\end{align*}
\end{definition}
\begin{definition}\label{def:retrigger} {\underline{Re-Trigger on Failure}}:

\begin{align*}
\hspace*{2.6cm}\textbf{G:}\{
PC \in TCB \rightarrow [~(\neg irq \land \neg dma_{en} \land PC \in TCB) \quad \textbf{W} \quad (PC = TCB_{max} \lor \textbf{F}(PC = TCB_{min})~]\}
\end{align*}
\end{definition}
\end{mdframed}
\vspace{-3mm}
\caption{Formal Specification of \acro end-to-end goals.}\label{fig:e2e-properties}
\end{figure*}

Using the notation from Section~\ref{sec:MCU_assumptions}, we proceed with the formal specification of \acro 
end-goals in LTL. Definition~\ref{def:trigger} specifies the ``guaranteed trigger'' property. It states in LTL that,
whenever a \trigger occurs, a TCB execution/invocation (starting at the legal entry point) will follow.

While Definition~\ref{def:trigger} guarantees that a particular interrupt of interest (\trigger) will cause the TCB execution, 
it does not guarantee proper execution of the TCB code as a whole. The ``re-trigger on failure'' property (per 
Definition~\ref{def:retrigger}) stipulates that, whenever TCB starts execution (i.e., $PC \in TCB$), it must  
execute without interrupts or DMA interference~\footnote{Since DMA could tamper with intermediate state/results 
in $DMEM$.}, i.e., $\neg irq \land \neg dma_{en} \land PC \in TCB$. This condition must hold until:
\begin{enumerate}
\item $PC=TCB_{max}$: the legal exit of TCB is reached, i.e., execution concluded successfully.
\item $F(PC = TCB_{min})$: another TCB execution (from scratch) has been triggered to occur.
\end{enumerate}
In other words, this specification reflects a cyclic requirement: either the security properties of the TCB proper 
execution are not violated, or TCB execution will re-start later.

Note that we use the quantifier Weak Until ({\bf W}) instead regular Until ({\bf U}), because, for some 
embedded applications, the TCB code may execute indefinitely; see Section~\ref{sec:app_actuator} for one such example.

\subsection{\acro Sub-Properties}\label{sec:sub_properties}
\begin{figure*}
\begin{mdframed}
\linespread{0.5}
\footnotesize
\begin{definition}{LTL Sub-Properties implemented and enforced by \acro.}\label{def:LTL_props}~\\

\textbf{Trusted $PMEM$ Updates:}

\begin{align}\label{ltl:trusted_updates_only}
\begin{split}
\text{\bf G}: \ \{
[\neg (PC \in TCB) \land W_{en} \land (D_{addr} \in PMEM)] \lor [DMA_{en} \land (DMA_{addr} \in PMEM)] \rightarrow reset \}
\end{split}
\end{align}
\textbf{IRQ Configuration Protection:}

\begin{align}\label{ltl:irq_cfg_protection}
\begin{split}
\hspace{0.8cm} \text{\bf G}: \ \{
[\neg (PC \in TCB) \land W_{en} \land (D_{addr} \in IRQ_{cfg})] \lor [DMA_{en} \land (DMA_{addr} \in IRQ_{cfg})] \rightarrow reset \}
\end{split}
\end{align}
\textbf{Interrupt Disablement Protection:}

\begin{align}\label{ltl:no_irq_disable}
\begin{split}
& \text{\bf G}: \ \{\neg reset \land gie \land \neg \textbf{X}(gie) \rightarrow (\textbf{X}(PC) \in TCB) \lor \textbf{X}(reset) \}
\end{split}
\end{align}
\textbf{TCB Execution Protection:}

\begin{align}\label{ltl:tcb_exec1}
\begin{split}
\text{\bf G}: \ \{\neg reset \land (PC \in TCB) \land \neg (\text{\bf X}(PC) \in TCB) \rightarrow PC = TCB_{max} \lor ~ \text{\bf X}(reset)\ \}
\end{split}
\end{align}

\begin{align}\label{ltl:tcb_exec2}
\begin{split}
& \text{\bf G}: \ \{\neg reset \land \neg (PC \in TCB) \land (\text{\bf X}(PC) \in TCB) \rightarrow \text{\bf X}(PC) = TCB_{min} \lor ~ \text{\bf X}(reset) \}
\end{split}
\end{align}

\begin{align}\label{ltl:tcb_exec3}
\begin{split}
& \text{\bf G}: \ \{(PC \in TCB) \land (irq \lor dma_{en}) \rightarrow reset \}
\end{split}
\end{align}
\end{definition}
\end{mdframed}
\vspace{-3mm}
\caption{Formal specification of sub-properties verifiably implemented by \acro hardware module.}\label{fig:sub-properties}
\end{figure*}

Based on our machine model and \acro end goals, we now postulate a set of necessary sub-properties to be implemented 
by \acro. Next, Section~\ref{sec:proof_composition} shows that this minimal set of sub-properties suffices to achieve \acro
end-to-end goals with a computer-checked proof. LTL specifications of the sub-properties are presented in 
Figure~\ref{fig:sub-properties}.

\acro enforces that only trusted updates are allowed to $PMEM$. 
\acro hardware issues a system-wide MCU $reset$ upon 
detecting any attempt to modify $PMEM$ at runtime, unless this modification comes from the execution of the TCB code itself.
This property is formalized in LTL~(\ref{ltl:trusted_updates_only}). It prevents any untrusted 
application software from misconfiguring IRQ-Table and interrupt handlers, as well as from modifying the $INIT$ segment and the 
TCB code itself, because these sections are located within $PMEM$. As a side benefit, it also prevents attacks that attempt 
to physically wear off Flash (usually used to implement $PMEM$ in low-end devices) by excessively and repeatedly 
overwriting it at runtime. Similarly, \acro prevents untrusted components from modifying $IRQ_{cfg}$ -- $DMEM$ 
registers controlling the \trigger configuration. This is specified by LTL~\ref{ltl:irq_cfg_protection}.

LTL~\ref{ltl:no_irq_disable} enforces that interrupts can not be globally disabled by untrusted applications.
Since, each \trigger is based on interrupts, disablement of all interrupts would allow untrusted software to 
disable the \trigger itself, and thus the active behavior of \acro. This requirement is specified by checking 
the relation between current and next values of $gie$, using the LTL ne\textbf{X}t operator. In order to switch 
$gie$ from logical $0$ (current cycle) to $1$ (next cycle), TCB must be executing when $gie$ becomes $0$ 
({\bf X}($PC) \in TCB$)), or the MCU will $reset$.

In order to assure that the TCB code is invoked and executed properly, \acro hardware implements 
LTL-s~(\ref{ltl:tcb_exec1}), (\ref{ltl:tcb_exec2}), and~(\ref{ltl:tcb_exec3}). LTL~\ref{ltl:tcb_exec1} enforces 
that the only way for $TCB$'s execution to terminate, without causing a $reset$, is through its last instruction 
(its only legal exit): $PC = TCB_{max}$. This is specified by checking the relation between current and next 
$PC$ values using LTL ne\textbf{X}t operator. If the current $PC$ value is within $TCB$, and next $PC$ value 
is outside $TCB$, then either current $PC$ value must be the address of $TCB_{max}$, or $reset$ is set to 
$1$ in the next cycle. Similarly, LTL~\ref{ltl:tcb_exec2} enforces that the only way for $PC$ to enter $TCB$ is through 
the very first instruction: $TCB_{min}$. This prevents $TCB$ execution from starting at some point in the middle 
of $TCB$, thus making sure that $TCB$ always executes in its entirety.  Finally, LTL~\ref{ltl:tcb_exec3} enforces 
that $reset$ is always set if interrupts or DMA modifications happen during $TCB$'s execution. Even though 
LTLs~\ref{ltl:tcb_exec1} and~\ref{ltl:tcb_exec2} already enforce that PC can not change to anywhere outside $TCB$, 
interrupts could be programmed to return to an arbitrary instruction within the $TCB$. Or, DMA could change 
$DMEM$ values currently in use by TCB. Both of these events can alter TCB behavior and are treated as violations.

Next, Section~\ref{sec:proof_composition} presents a computer-checked proof for the sufficiency of this set of sub-properties 
to imply \acro end-to-end goals. Then, Section~\ref{sec:verified_impl} presents FSM-s from our Verilog implementation, that are 
formally verified to correctly implement each of these requirements.

\subsection{\acro Composition Proof}\label{sec:proof_composition}
\acro end-to-end sufficiency is stated in Theorems~\ref{th:trigger} and~\ref{th:re-trigger}.
\begin{figure}
\begin{mdframed}
\footnotesize
\begin{theorem}\label{th:trigger}
$\text{Definition~\ref{def:machine_model}} \land \text{LTLs \ref{ltl:trusted_updates_only},\ref{ltl:irq_cfg_protection},\ref{ltl:no_irq_disable}} \rightarrow \text{Definition~\ref{def:trigger}}$.
\end{theorem}
\begin{theorem}\label{th:re-trigger}
$\text{Definition~\ref{def:machine_model}} \land \text{LTLs~\ref{ltl:trusted_updates_only},\ref{ltl:tcb_exec1},\ref{ltl:tcb_exec2},\ref{ltl:tcb_exec3}} \rightarrow \text{Definition~\ref{def:retrigger}}$.
\end{theorem}
\end{mdframed}
\end{figure}
The complete computer-checked proofs (using Spot2.0~\cite{spot}) of Theorems~\ref{th:trigger} and~\ref{th:re-trigger} 
are publicly available at~\cite{repo}. Below we present the intuition behind them.

\begin{proof}[Proof of Theorem~\ref{th:trigger} (Intuition)]
From machine model's LTL~(\ref{ltl:no_mod_no_prob}), as long as the (1) initial trigger configuration is never modified from 
outside the $TCB$; and (2) interrupts are never disabled from outside the $TCB$; it follows that a \trigger will cause 
a proper invocation of the TCB code. Also, successful modifications to the \trigger's configuration imply writing to 
$PMEM$ or $IRQ_{cfg}$ without causing a $reset$ (per LTL (\ref{ltl:trigger_mod})). Since \acro verified implementation 
guarantees that memory modifications (specified in LTL~(\ref{ltl:mem_write})) to $PMEM$ (LTL~(\ref{ltl:trusted_updates_only})) 
or $IRQ_{cfg}$ (LTL~(\ref{ltl:irq_cfg_protection})) always cause a $reset$, illegal modifications to $trigger_{cfg}$ are 
never successful. Finally, LTL~(\ref{ltl:no_irq_disable}) assures that any illegal interrupt disablement always causes 
a $reset$, and is thus never successful). Therefore, \acro satisfies all necessary conditions to guarantee the goal 
in Definition~\ref{def:trigger}.
\end{proof}

\begin{proof}[Proof of Theorem~\ref{th:re-trigger} (Intuition)]
The fact that a reset always causes a later call to the TCB follows from the machine model's 
LTL~(\ref{ltl:no_mod_no_prob_reset}) and \acro guarantee in LTL~(\ref{ltl:trusted_updates_only}).
LTLs~(\ref{ltl:tcb_exec1}) and ~(\ref{ltl:tcb_exec1}) ensure that the TCB executable is properly invoked 
and executes atomically, until its legal exit. Otherwise a $reset$ flag is set, which (from the above 
argument) implies a new call to TCB. Finally, LTL~\ref{ltl:tcb_exec3} assures that any interrupt or 
DMA activity during TCB execution will cause a $reset$, thus triggering a future TCB call and 
satisfying Definition~\ref{def:retrigger}.
\end{proof}
See~\cite{repo} for the formal computer-checked proofs.

\subsection{Sub-Module Implementation+Verification}\label{sec:verified_impl}
Following the sufficiency proof in Section~\ref{sec:proof_composition} for sub-properties in 
Definition~\ref{def:LTL_props}, we proceed with the implementation and formal verification of \acro hardware
using the NuSMV model-checker (see Section~\ref{sec:verif_background} for details).

\acro modules are implemented as Mealy FSMs (where outputs change with the current state and current inputs) 
in Verilog. Each FSM has one output: a local $reset$. \acro output $reset$ is given by the disjunction (logic {\it or}) 
of local $reset$-s of all sub-modules. Thus, a violation detected by any sub-module causes \acro to trigger an 
immediate MCU $reset$. For the sake of easy presentation we do not explicitly represent the value of 
$reset$ in the figures. Instead, we define the following implicit representation:
\begin{compactenum}
\item $reset$ output is 1 whenever an FSM transitions to the $RESET$ state (represented in red color);
\item $reset$ output remains 1 until a transition leaving the $RESET$ state is triggered;
\item $reset$ output is 0 in all other states (represented in blue color).
\end{compactenum}
Note that all FSM-s remain in the $RESET$ state until $PC=0$, which signals that the MCU reset routine finished.

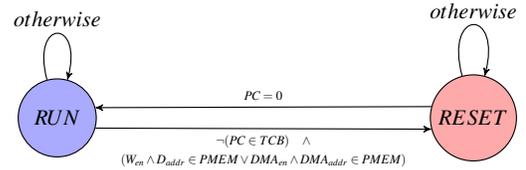
\begin{figure}
\begin{center}
\noindent\resizebox{0.8\columnwidth}{!}{%
\begin{tikzpicture}[->,>=stealth',auto,node distance=8.0cm,semithick]
\tikzstyle{every state}=[minimum size=1.5cm]
\tikzstyle{every node}=[font=\large]

\node[state, fill={rgb:blue,1;white,2}] (A){$RUN$};
\node[state, fill={rgb:red,1;white,2}]         (B) [right of=A,align=center]{$RESET$};

\path[->,every loop/.style={looseness=8}] 
(A) edge [loop above] node {$otherwise$} (A)
(B) edge [loop above] node {$otherwise$} (B);

\draw[->] (A.345) -- node[rotate=0,below, align=center,auto=right] {\scriptsize \shortstack{$\neg(PC \in TCB) \quad \land $\\$ (W_{en} \land D_{addr} \in PMEM \lor DMA_{en} \land DMA_{addr} \in PMEM)$}} (B.195);
\draw[<-] (A.15) -- node[rotate=0,above] {\scriptsize \shortstack{$PC=0$}} (B.165);
\end{tikzpicture}
}
\vspace{-1mm}
\caption{Verified FSM for LTL~\ref{ltl:trusted_updates_only}.}
\label{fig:no_mod_PMEM}
\end{center}
\end{figure}

Figure~\ref{fig:no_mod_PMEM} illustrates \acro sub-module responsible for assuring that $PMEM$ modifications 
are only allowed from within the TCB. This minimal 2-state machine works by monitoring $PC$, $W_{en}$, 
$D_{addr}$, $DMA_{en}$, and $DMA_{addr}$ to detect illegal modification attempts by switching from $RUN$ 
to $RESET$ state, upon detection of any such action. It is verified to adhere to LTL~(\ref{ltl:trusted_updates_only}).
A similar FSM is used to verifiably enforce LTL~(\ref{ltl:irq_cfg_protection}), with the only distinction of checking for 
writes within $IRQ_{cfg}$ region instead, i.e., $D_{addr} \in IRQ_{cfg})$ and $DMA_{addr} \in IRQ_{cfg})$. 
We omit the illustration of this FSM to conserve space, due to page limits. 

Figure~\ref{fig:no_irq_disable} presents an FSM implementing LTL~\ref{ltl:no_irq_disable}.
It monitors the ``global interrupt enable'' ($gie$) signal to detect attempts to illegally disable interrupts. 
It consists of three states: (1) $ON$, representing execution periods where $gie=1$; (2) $OFF$, for cases 
where $gie=0$, and (3) $RESET$. To switch between $ON$ and $OFF$ states, this FSM requires $PC \in TCB$, 
thus preventing misconfiguration by untrusted software.

\begin{figure}
\begin{center}
\noindent\resizebox{0.6\columnwidth}{!}{%
\begin{tikzpicture}[->,>=stealth',auto,node distance=4.0cm,semithick]
\tikzstyle{every state}=[minimum size=1.5cm]
\tikzstyle{every node}=[font=\large]

\node[state,fill={rgb:red,1;white,2}] (A){$RESET$};
\node[state,fill={rgb:blue,1;white,2}]         (B) [above of=A,left of=A, align=center]{$~OFF~$};
\node[state,fill={rgb:blue,1;white,2}]         (C) [above of=A,right of=A ,align=center]{$~ON~$};

\path[->,every loop/.style={looseness=8}] (A) edge [bend right=10]  node [right] {\small $PC=0$} (B)
edge [loop below] node {\small$otherwise$} (A)
(B) edge [loop above] node {$\neg~gie$} (C)
edge [bend right=10] node [below] {\small$gie \land PC \in TCB $} (C)
edge [bend right=10] node [left] {\small$otherwise$} (A)
(C) edge [loop above] node {\shortstack{$gie$}} (C)
edge [bend right=10] node [above] {\small$\neg gie \land PC \in TCB $} (B)
edge node {\small$otherwise$} (A);
\end{tikzpicture}
}
\vspace{-1mm}
\caption{Verified FSM for LTL~\ref{ltl:no_irq_disable}.}
\label{fig:no_irq_disable}
\end{center}
\end{figure}
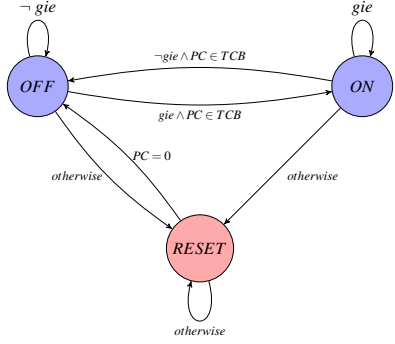

Finally, the FSM in Figure~\ref{fig:atomicity_fsm} verifiably implements LTL-s~\ref{ltl:tcb_exec1},~\ref{ltl:tcb_exec2}, 
and~\ref{ltl:tcb_exec3}. This FSM has 5 states, one of which is $RESET$. Two basic states correspond to whenever: 
the TCB is executing (state ``$\in TCB$''), and not executing (state ``$\notin TCB$''). From $\notin TCB$ the only 
reachable path to $\in TCB$ is through state $TCB_{entry}$, which requires $PC=TCB_{min}$ -- TCB only legal entry point.
Similarly, from $\in TCB$ the only reachable path to $\notin TCB$ is through state $TCB_{exit}$, which requires 
$PC=TCB_{max}$ -- TCB only legal exit. Also, in all states where $PC \in TCB$ (including entry and exit transitions) 
this FSM requires DMA and interrupts to remain inactive. Any violation of these requirements, in any of the four regular 
states, causes the FSM transition to $RESET$, thus enforcing TCB execution protection.

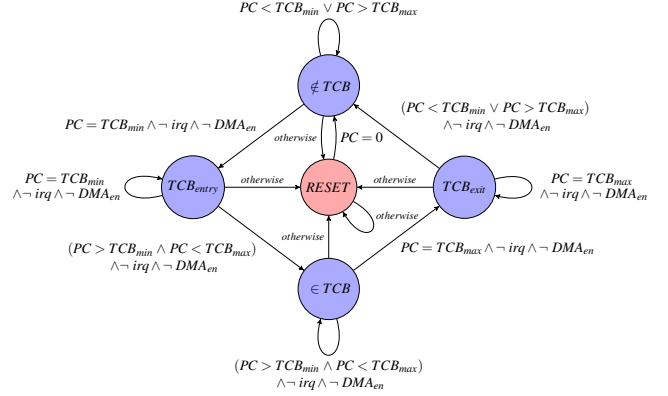
\begin{figure}
\begin{center}
\noindent\resizebox{\columnwidth}{!}{%
\begin{tikzpicture}[->,>=stealth',auto,node distance=4.0cm,semithick]
\tikzstyle{every state}=[minimum size=1.5cm]
\tikzstyle{every node}=[font=\large]

\node[state,fill={rgb:red,1;white,2}] (A){$RESET$};
\node[state,fill={rgb:blue,1;white,2}]         (B) [above of=A,align=center, yshift=-1cm]{$~\notin TCB~$};
\node[state,fill={rgb:blue,1;white,2}]         (C) [left  of=A,align=center]{$TCB_{entry}$};
\node[state,fill={rgb:blue,1;white,2}]         (E) [below of=A, yshift=1cm] {$~\in TCB~$};
\node[state,fill={rgb:blue,1;white,2}]         (D) [right of=A]{$~TCB_{exit}~$};

\path[->,every loop/.style={looseness=8}] (A) edge [bend right=10]  node [right] {$PC=0$} (B)
edge [out=330,in=300,looseness=8] node[above right, yshift=.2cm] {\small$otherwise$} (A)
(B)  edge [loop above] node {$PC<TCB_{min}\,\lor\,PC>TCB_{max}$} (C)
edge node [above left] {$PC=TCB_{min}\land \neg~irq \land \neg~DMA_{en}$} (C)
edge [bend right=10] node [left] {\small$otherwise$} (A)
(C)  edge [loop left] node [above right, left] {\shortstack{$PC=TCB_{min}$\\~$\land \neg~irq \land \neg~DMA_{en}$}} (C)
edge node [below left] {\shortstack{$(PC>TCB_{min}\,\land\,PC<TCB_{max})$\\~$\land \neg~irq \land \neg~DMA_{en}$}} (E)
edge node {\small$otherwise$} (A)
(E)  edge [loop below] node {\shortstack{$(PC>TCB_{min}\,\land\,PC<TCB_{max})$\\~$\land \neg~irq \land \neg~DMA_{en}$}} (C)
edge node [below right] {$PC=TCB_{max}\land \neg~irq \land \neg~DMA_{en}$} (D)
edge node [left] {\small$otherwise$} (A)
(D)  edge [loop right] node [above left, right]  {\shortstack{$PC=TCB_{max}$\\~$\land \neg~irq \land \neg~DMA_{en}$}} (D)
edge node [above right] {\shortstack{$(PC<TCB_{min}\,\lor\,PC>TCB_{max})$\\~$\land \neg~irq \land \neg~DMA_{en}$}} (B)
edge node [above]  {\small$otherwise$} (A);
\end{tikzpicture}
}
\vspace{-1mm}
\caption{Verified FSM for LTLs~\ref{ltl:tcb_exec1}--\ref{ltl:tcb_exec3}.}
\label{fig:atomicity_fsm}
\end{center}
\end{figure}

\subsection{TCB Confidentiality}\label{sec:confidentiality}
One instance of \acro enables confidentiality of TCB data and code with respect to 
untrusted applications. This is of particular interest when \func implements cryptographic 
functions or privacy-sensitive tasks. 

This goal can be achieved by including and epilogue phase in the TCB executable, with the goal of 
performing a $DMEM$ cleanup, erasing all traces of the TCB execution from the stack and heap.
While the TCB execution may be interrupted before the execution of the epilogue phase, such an 
interruption will cause an MCU $reset$. The \textit{Re-Trigger on Failure} property assures that 
TCB code will execute (as a whole) after any $reset$ and will thus erase remaining execution 
traces from $DMEM$ before subsequent execution of untrusted applications.
In a similar vein, if confidentiality of the executable is desirable, it can be implemented following 
LTL~(\ref{ltl:confidentiality}), which formalizes read attempts based on $R_{en}$ signal:
\begin{align}\label{ltl:confidentiality}
\small
\begin{split}
&\text{\bf G}: \ \{ \\
&\quad [\neg (PC \in TCB) \land R_{en} \land (D_{addr} \in TCB) \lor \\
&\quad DMA_{en} \land (DMA_{addr} \in TCB)] \rightarrow reset\\
&\}
\end{split}
\end{align}
An FSM implementing this property is shown in Figure~\ref{fig:FSM_confidentiality}. Note that, despite 
visual similarity with the FSM in Figure~\ref{fig:no_mod_PMEM}, the confidentiality 
FSM checks for \emph{reads} (instead of writes) to the TCB (instead of entire $PMEM$).

\begin{figure}
\begin{center}
\noindent\resizebox{0.8\columnwidth}{!}{%
\begin{tikzpicture}[->,>=stealth',auto,node distance=8.0cm,semithick]
\tikzstyle{every state}=[minimum size=1.5cm]
\tikzstyle{every node}=[font=\large]

\node[state, fill={rgb:blue,1;white,2}] (A){$RUN$};
\node[state, fill={rgb:red,1;white,2}]         (B) [right of=A,align=center]{$RESET$};

\path[->,every loop/.style={looseness=8}] 
(A) edge [loop above] node {$otherwise$} (A)
(B) edge [loop above] node {$otherwise$} (B);

\draw[->] (A.345) -- node[rotate=0,below, align=center,auto=right] {\scriptsize \shortstack{$\neg(PC \in TCB) \quad \land $\\$ (R_{en} \land D_{addr} \in TCB \lor DMA_{en} \land DMA_{addr} \in TCB)$}} (B.195);
\draw[<-] (A.15) -- node[rotate=0,above] {\scriptsize \shortstack{$PC=0$}} (B.165);
\end{tikzpicture}
}
\vspace{-1mm}
\caption{Verified FSM for LTL~\ref{ltl:confidentiality}.}
\label{fig:FSM_confidentiality}
\end{center}
\end{figure}
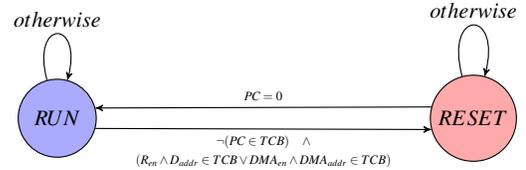

This property prevents external reads to the TCB executable by monitoring $R_{en}$, $D_{addr}$, and DMA. 
When combined with the aforementioned erasure epilogue, it also enables secure storage of cryptographic 
secrets within the TCB binary (as in architectures such as~\cite{verify_and_revive,simple, hydra}).
This part of \acro design is optional, since some embedded applications do not require confidentiality,
e.g., those discussed in Sections~\ref{sec:app_actuator} and~\ref{sec:app_scheduling}.

\subsection{Resets \& Availability}\label{sec:availability}
One important remaining issue is {\it availability}. For example, malware might interrupt (or tamper with) 
with $INIT$ execution after a $reset$ preventing the subsequent execution of TCB. Also, malware could
to interrupt the TCB execution, after each \emph{re-\trigger}, with the goal of resetting the MCU indefinitely, 
and thereby preventing TCB execution from ever completing its task.

We observe that such actions are not possible, since they would require either DMA activity or interrupts to: 
(1) hijack $INIT$ control-flow; or (2) abuse \acro to successively $reset$ the MCU during TCB 
execution after each re-trigger. Given {\bf H5} interrupts are disabled by default at boot time.
Additionally, {\bf H4} states that any prior DMA configuration is cleared to the default disabled state after a $reset$. 
Hence, $INIT$ and the first execution of TCB after a $reset$ cannot be interrupted or tampered with by DMA.

Finally, we note that, despite preventing security violations by (and implementing re-trigger based on) resetting 
the MCU, \acro does not provide any advantage to malware that aims to simply disrupt execution of (non-TCB) 
applications by causing $resets$. Any software running on bare metal (including malware) can always intentionally 
reset the MCU. Resets are the default mechanism to recover from regular software faults 
on unmodified (off-the-shelf) low-end MCU-s, regardless of \acro.

\section{Sample Applications}\label{sec:applications}
Many low-end MCU use-cases and applications can benefit from \trigger-based active RoTs.
To demonstrate generality of \acro, we prototyped three concrete examples, each with a different type of \trigger-s.
This section overviews these examples: (1) GPIO-TCB uses external analog events (Section~\ref{sec:app_actuator}), 
TimerTCB uses timers (Section~\ref{sec:app_scheduling}), and NetTCB uses network events (Section~\ref{sec:app_nettcb}). 
Finally, Section~\ref{sec:app_authWDT} discusses how \acro can match active security services 
proposed in~\cite{proactive1} and~\cite{proactive2}.

\subsection{GPIO-TCB: Critical Sensing+Actuation}\label{sec:app_actuator}
The first example, GPIO-TCB, operates in the context of a safety-critical temperature sensor.
We want to use \acro to assure that the sensor's most safety-critical function -- \emph{sounding an alarm} -- is 
never prevented from executing due to software compromise of the underlying MCU.
We use a standard built-in MCU interrupt, based on General Purpose Input/Output (GPIO) to implement \trigger. 
Since this is our first example, we discuss GPIO-TCB in more detail than the other two.

\begin{figure}
\begin{lstlisting}[basicstyle=\tiny, numberstyle=\tiny]
int main() {
	TCB(0);
	main_loop();
	return 0; 
}
\end{lstlisting}
\vspace{-0.5cm}
\caption{Program Entry Point}\label{lst:main}
\end{figure}

As shown in Figure~\ref{lst:main}, MCU execution always starts by calling the TCB (at line 2). Therefore,
after MCU initialization/reset, unprivileged (non-TCB) applications can only execute after the TCB; assuming, 
of course, that formal guarantees discussed in Section~\ref{sec:core} hold. These applications are 
implemented inside \textit{main\_loop} function (at line 3).

The correct \trigger configuration in GPIO-TCB can be achieved in two ways. The first way is to set $IRQ_{cfg}$ to 
the desired parameters at MCU deployment time, by physically writing this configuration to $IRQ_{cfg}$.
The second option is to implement this configuration in software as a part of the TCB. Since the TCB 
is always the first to run after initialization/reset, it will configure $IRQ_{cfg}$ correctly, enabling 
subsequent \trigger-s at runtime.

\begin{figure}
\begin{lstlisting}[basicstyle=\tiny, numberstyle=\tiny]
void setup (void) {
	P1DIR  = 0x00;
	P1IE   = 0x01;
	P1IES  = 0x00;
	P1IFG  = 0x00;
}
\end{lstlisting}
\vspace{-0.5cm}
\caption{Trigger Setup}\label{lst:trigger}
\end{figure}

Figure~\ref{lst:trigger} exemplifies $IRQ_{cfg}$ configuration, implemented as part of the TCB, i.e., called from within 
the TCB. This \texttt{setup} function is statically linked to be located inside the TCB memory region, thus respecting 
``TCB Execution Protection'' LTL rules (see Definition~\ref{def:LTL_props}). This $IRQ_{cfg}$ setup first configures 
the physical port $P1$ as an input (line 2, ``P1 direction'' set to 0x00, whereas 0x01 would set it as an output). 
At line 3, $P1$ is set as ``interrupt-enabled'' ($P1IE=0x01$). A value of $P1IES  = 0x00$ (line 4) indicates that, 
if the physical voltage input of $P1$ changes from logic $0$ to $1$ (``low-to-high` transition), a GPIO interrupt will 
be triggered and the respective handler will be called. Finally, $P1IFG$ is cleared to indicate that the MCU is 
free to receive interrupts (as opposed to busy). We note that this initial trusted configuration of $IRQ_{cfg}$ cannot 
be modified afterwards by untrusted applications due to \acro guarantees (see Section~\ref{sec:core}).
Based on this configuration, an analog temperature sensing circuit (i.e., a voltage divider implemented using a 
thermistor (i.e., a resistance thermometer -- a resistor whose resistance varies with temperature) 
is connected to port P1. Resistances in this circuit are set to achieve $5 V$ (logic $1$) when temperature 
exceeds a fixed threshold, thus triggering a $P1$ interrupt.

P1 interrupt is handled by the function in Figure~\ref{lst:interrupt_gpio}. This is configured using the 
\texttt{$interrupt(PORT1\_VECTOR)$} macro. This handler essentially calls \acro TCB. Parameter 
$1$ in the TCB call distinguishes a regular TCB call from a TCB call following initialization/reset.

\begin{figure}
\begin{lstlisting}[basicstyle=\tiny, numberstyle=\tiny]
interrupt(PORT1_VECTOR) port1_isr(void) {
	TCB(1);
} 
\end{lstlisting}
\vspace{-0.5cm}
\caption{GPIO Handling Routine}\label{lst:interrupt_gpio}
\end{figure}

\begin{figure}
\begin{lstlisting}[basicstyle=\tiny, numberstyle=\tiny]
TCB (uint8_t init) {
	dint();
	if (!init) {
		setup();
	}
	volatile uint_64 i=0;
	P3DIR = 0x01;
	P3OUT = 0x01;
	while (i<100000000) i++;
	P3OUT = 0x00;
	eint();
	return();
}
\end{lstlisting}
\vspace{-0.5cm}
\caption{$IRQ_{cfg}$ initialization}\label{lst:tcb}
\end{figure}

Figure~\ref{lst:tcb} depicts the TCB implementation of \func. Once triggered, TCB disables interrupts (\texttt{dint}), 
calls \texttt{setup} (if this is the first TCB call after initialization/reset), and activates GPIO port P3 for a 
predefined number of cycles. P3 is connected to a buzzer (a high frequency oscillator 
circuit used for generating a buzzing sound), guaranteeing that the 
alarm will sound. Upon completion, TCB re-enables interrupts and returns control to the regular application(s).

Note that, as discussed in Section~~\ref{sec:core}, executables corresponding to Figures~\ref{lst:interrupt_gpio} 
and~\ref{lst:tcb} are also protected by \acro. Thus, their behavior cannot be modified by untrusted/compromised software.

\subsection{TimerTCB: Secure Real-Time Scheduling}\label{sec:app_scheduling}
The second example of \acro, TimerTCB, is in the domain of real-time task scheduling. 
Without \acro (even in the presence of a passive RoT), a compromised MCU controlled by malware could ignore performing
its periodic security- or safety-critical tasks. (Recall that targeted MCU-s typically run bare-metal software, with no OS 
support for preempting tasks). We show how \acro can ensure that a prescribed task, implemented within the TCB 
periodically executes.

Unlike our first example in Section~\ref{sec:app_actuator}, TimerTCB only requires modifying $IRQ_{cfg}$, as illustrated in 
Figure~\ref{lst:timer}. This shows the relative ease of use of \acro. The \textit{setup} function is modified to enable the MCU's 
built-in timer to cause interrupts (at line 2). Interrupts are set to occur whenever the timer's counter reaches a desired value 
(at line 3). The timer is set to increment the counter with edges of a particular MCU clock ($MC1$, at line 4).
As in the first example, the corresponding interrupt handler is set to always call the TCB (Figure~\ref{lst:interrupt_timer}). 
In turn, the TCB can implement \func as an arbitrary safety-critical periodic task.

\begin{figure}
\begin{lstlisting}[basicstyle=\tiny, numberstyle=\tiny]
void setup (void) {
	CCTL0 = CCIE;
	CCR0  = 1000000;
	TACTL = TASSEL_2 + MC_1;
}
\end{lstlisting} \vspace{-0.5cm}
\caption{Timer Trigger Setup}\label{lst:timer}
\end{figure}

\begin{figure}
\begin{lstlisting}[basicstyle=\tiny, numberstyle=\tiny]
interrupt(TIMERA0_VECTOR) timera_isr(void) {
	TCB(1)
} 
\end{lstlisting} \vspace{-0.5cm}
\caption{Timer Handle Routine}\label{lst:interrupt_timer}
\end{figure}

\subsection{NetTCB: Network Event-based \trigger}\label{sec:app_nettcb}
The last example, NetTCB, uses network event-based \trigger to ensure that the TCB 
quickly filters all received network packets to identify those that carry TCB-destined commands 
and take action. Incoming packets that do not contain such commands are ignored by the TCB 
and passed on to applications through the regular interface (i.e., reading from the UART buffer). 
In this example, we implement guaranteed receipt of external $reset$ commands from some trusted
remote entity. This functionality might be desirable after an MCU malfunction (e.g., due to a deadlock) is detected.

In NetTCB, \trigger is configured to trap network events. $IRQ_{cfg}$ is set such that each incoming UART message
causes an interrupt, as shown in Figure~\ref{lst:uart}. The TCB implementation, shown in Figure~\ref{lst:nettcb}, filters 
messages based their initial character $'r'$ which is predefined as a command to $reset$ the MCU.
{\bf Note that:} in practice such critical commands should be authenticated by the TCB. Although this authentication  
should be implemented within the TCB, we omit it from this discussion for the sake of simplicity, and refer
to \cite{brasser2016remote} for a discussion of authentication of external requests in this setting.

\begin{figure}
\begin{lstlisting}[basicstyle=\tiny, numberstyle=\tiny]
void setup (void) {
	UART_BAUD = BAUD;
	UART_CTL  = UART_EN | UART_IEN_RX;
}
\end{lstlisting} \vspace{-0.5cm}
\caption{UART Trigger Setup}\label{lst:uart}
\end{figure}

\begin{figure}
\begin{lstlisting}[basicstyle=\tiny, numberstyle=\tiny]
wakeup interrupt (UART_RX_VECTOR) INT_uart_rx(void) {
	TCB(1);
}

TCB (uint8_t init) {
	dint();
	if (!init) {
		setup();
	}
	rxdata = UART_RXD;
	if (rxdata == 'r') {
		reset();
	}
	eint();
	return();
}
\end{lstlisting} \vspace{-0.5cm}
\caption{NetTCB Handler Routine and TCB Implementation}\label{lst:nettcb}
\end{figure}

\subsection{Comparison with~\cite{proactive1} and~\cite{proactive2}}\label{sec:app_authWDT}
Recent work proposed security services that can be interpreted as active RoT-s. However, these efforts 
aimed at higher-end embedded devices and require substantial hardware support: Authenticated 
Watchdog Timer (AWDT) implemented as a separate (stand-alone) microprocessor~\cite{proactive1},
or ARM TrustZone~\cite{proactive2}. Each requirement is, by itself, far more expensive than the cost 
of a typical low-end MCU targeted in this paper (see Section~\ref{sec:scope}).

In terms of functionality, both \cite{proactive1} and \cite{proactive2} are based on timers. They use 
AWDT to force a reset of the device. As in \acro, in these designs, the TCB is the first code to execute;
this property is referred to as ``gated boot'' in~\cite{proactive1}. However, unlike~\acro, ~\cite{proactive1,proactive2} 
do not consider active RoT behavior obtainable from other types of interrupts, e.g., as in \acro examples  
in Sections~\ref{sec:app_actuator} and~\ref{sec:app_nettcb}). We believe that this is partly because 
these designs were originally intended as an active means to enforce memory integrity, rather than a 
general approach to guaranteed execution of trusted tasks based on arbitrary \trigger-s (as in \acro). 
Note that \acro design is general enough to realize an active means to enforce memory integrity. 
This can be achieved by incorporating an integrity-ensuring function (e.g, a suitable cryptographic keyed hash) 
into \acro TCB and using it to check PMEM state upon a timer-based \trigger. 

Finally, we emphasize that prior results involved neither formally specified designs nor 
formally verified open-source implementations. As discussed in Section~\ref{sec:intro}, 
we believe these features to be important for eventual adoption of this type of architecture.

\section{Implementation \& Evaluation}\label{sec:evaluation}

We prototyped \acro (adhering to its architecture in Figure~\ref{fig:architecture}) using an open-source 
implementation of the popular MSP430 MCU -- openMPS430~\cite{openmsp430} from OpenCores. 
In addition to \acro module, we reserve, by default, $2$ KBytes of PMEM for TCB functions.
This size choice is configurable at manufacturing time and MCU-s manufactured for different 
purposes can choose different sizes. In our prototype, $2$ KBytes is a a reasonable choice, 
corresponding to $5-25\%$ of the typical amount of PMEM in low-end MCU-s.
The prototype supports one \trigger of each type: timer-based, external hardware, and network. 
This support is achieved by implementing the $IRQ_{cfg}$ protection, as described in Section~\ref{sec:core}. 
The MCU already includes multiple timers and GPIO ports that can be selected to act as \trigger-s. 
By default, one of each is used by our prototype. This enables the full set of types of applications 
discussed in Section~\ref{sec:applications}.

As a proof-of-concept, we use Xilinx Vivado to synthesize our design and deploy it using the Basys3 Artix-7 FPGA board. Prototyping using FPGAs is common in both research and industry. Once a hardware design is synthesizable in an FPGA, the same design can be used to manufacture an Application-Specific Integrated Circuit (ASIC) at larger scale.
\begin{center}
\vspace{-2mm}
\textbf{Hardware \& Memory Overhead}
\vspace{-2mm}
\end{center}
Table~\ref{tab:hw} reports \acro hardware overhead as compared to unmodified OpenMSP430~\cite{openmsp430}. 
Similar to the related work~\cite{sancus,apex,vrased,litehax,cflat,lofat,atrium}, we consider hardware overhead in 
terms of additional Look-Up Tables (LUT-s) and registers. The increase in the number of LUT-s can be used as 
an estimate of the additional chip cost and size required for combinatorial logic, while the number of registers 
offers an estimate on the memory overhead required by the sequential logic in \acro FSMs.

\acro hardware overhead is small with respect to the unmodified MCU core -- it requires 2.3\% and 4.8\% 
additional LUT-s and registers, respectively. In absolute numbers, \acro adds 33 registers and 42 LUT-s to the underlying MCU.
\begin{center}
\vspace{-2mm}
\textbf{Runtime \& Memory Overhead}
\vspace{-2mm}
\end{center}
We observed no discernible overhead for software execution time on the \acro-enabled MCU. 
This is expected, since \acro introduces no new instructions or modifications to the MSP430 ISA and 
to the application executables. \acro hardware runs in parallel with the original MSP430 CPU.
Aside from the reserved PMEM space for storing the TCB code, \acro also does not incur any memory overhead. This behavior does not depend on the number of functions or triggers used inside the TCB. 
\begin{center}
\vspace{-2mm}
\textbf{Verification Cost}
\vspace{-2mm}
\end{center}
We verify \acro on an Ubuntu 18.04 machine running at 3.40GHz. Results are also shown in Table~\ref{tab:hw}.
\acro implementation verification requires checking 7 LTL statements. 
The overall verification pipeline (described in Section~\ref{sec:verif_background}) is fast enough to run on a commodity desktop in quasi-real-time.

\begin{table*}[!hbtp]
\footnotesize
\centering
\begin{tabular}{l|cc|c|cccc}
\hline
& \multicolumn{2}{c|}{Hardware} & Reserved & \multicolumn{4}{c}{Verification} \\
& Reg & LUT & PMEM/Flash (bytes) & \# LTL Invariants & Verified Verilog LoC & Time (s) & Mem (MB) \\ \hline\hline
\multicolumn{1}{l|}{OpenMSP430~\cite{openmsp430}}             &  692  & 1813 & 0 & - & - & - & -   \\
\multicolumn{1}{l|}{OpenMSP430 + \acro}   & 725   & 
1855 & 2048 (default) & 7 & 484 & 3.1 & 13.5  \\ \hline
\end{tabular}%
\caption{\acro Hardware overhead and verification costs.}
\label{tab:hw}
\end{table*}

\begin{center}
\vspace{-2mm}
\textbf{Comparison with Prior RoTs}
\vspace{-2mm}
\end{center}
To the best of our knowledge, \acro is the first active RoT targeting this lowest-end class of devices.
Nonetheless, to provide a overhead point of reference and a comparison, we contrast \acro's overhead with 
that of state-of-the-art \underline{passive} RoTs in the same class. We note that the results from~\cite{proactive1,proactive2} can not be compared to \acro quantitatively.
As noted in Section~\ref{sec:app_authWDT}, \cite{proactive1} relies on a standalone additional MCU and~\cite{proactive2} on ARM TrustZone. Both of these are (by themselves) more expensive and sophisticated than the entire MSP430 MCU (and similar low-end MCUs in the same class).
Our quantitative comparison focuses on VRASED~\cite{vrased}, APEX~\cite{apex}, 
and SANCUS~\cite{sancus}: passive RoTs implemented on the same MCU and thus directly comparable (cost-wise). Table~\ref{tab:functions} provides a qualitative comparison between the aforementioned relevant designs. Figure~\ref{fig:comparison} depicts the relative  overhead (in \%) of \acro, VRASED, APEX, and SANCUS with respect to the total hardware cost of the unmodified MSP430 MCU core.

In comparison with prior passive architectures, \acro presents lower hardware overhead. In part, this is due to the fact that it leverages interrupt hardware support already present in the underlying MCU to implement its triggers.
SANCUS presents substantially higher cost as it implements task isolation and a cryptographic hash engine (for the purpose of verifying software integrity) in hardware.
VRASED presents slightly higher cost than \acro. It also necessitates some properties that are similar to \acro's (e.g., access control to particular memory segments and atomicity of its attestation implementation). In addition, VRASED also requires hardware support for an exclusive stack in DMEM.
APEX hardware is a super-set of VRASED's, providing an additional proof of execution function in hardware. As such it requires strictly more hardware support, presenting slightly higher cost.
\acro also reserves approximately $3.1\%$ ($2$ KBytes) of the MCU-s 16-bit address space for storing the TCB code. This value is freely configurable, and chosen as a sensible default to support our envisioned RoT tasks (including sample applications in Section~\ref{sec:applications}). \acro-enabled MCUs manufactured for different use-cases could increase or decrease this amount accordingly.

\begin{table}
\resizebox{\linewidth}{!}{%
\begin{tabular}{||l||l|l|l|l||}
\hline
{\bf Architecture} & {\bf Behavior} & {\bf Service}                                   & {\bf HW Support}    & {\bf Verified?} \\ \hline
VRASED~\cite{vrased}                                          & Passive  & Attestation                        & RTL Design    & Yes                \\ \hline
SANCUS~\cite{sancus}                                          & Passive  & \begin{tabular}[c]{@{}l@{}}Attestation \&\\  Isolation\end{tabular}   & RTL Design    & No                 \\ \hline
APEX~\cite{apex}                                            & Passive  & \begin{tabular}[c]{@{}l@{}} Attestation \& \\ Proof of Execution \end{tabular} & RTL Design    & Yes                \\ \hline
Cider~\cite{proactive1}    & Active   & Timer-based trigger                   & Additional MCU & No                 \\ \hline
Lazarus\cite{proactive2} & Active   & Timer-based trigger                   & ARM TrustZone & No                 \\ \hline
{\bf \acro (this paper)}               & Active   & IRQ-based trigger            & RTL Design    & Yes                \\ \hline
\end{tabular}
}
\vspace{-1em}
\caption{Qualitative Comparison}\label{tab:functions}
\vspace{-1em}
\end{table}
%
%
%
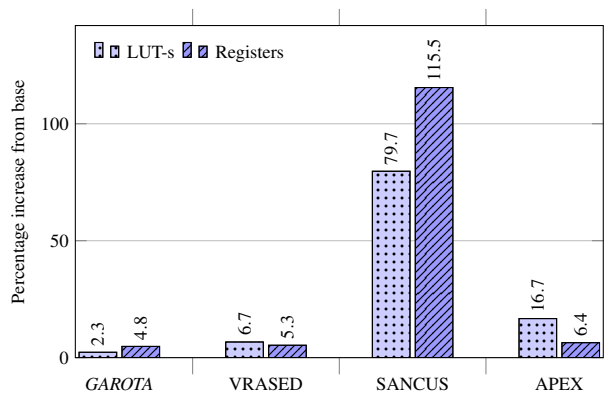
\begin{figure}[b]
\pgfplotstableread{
0  2.3        4.8
1  6.7        5.3
2 79.7      115.5
3 16.7        6.4
}\dataset
\begin{tikzpicture}
\scriptsize
\begin{axis}[ybar,
bar width=.5cm,	
width=8.6cm,	
height=6cm,	
ymin=0,	
ymax=142,	
ymajorgrids=true,       	
ylabel={Percentage increase from base},
xtick=data,
xticklabels = {
\strut \acro,
\strut VRASED,
\strut SANCUS,
\strut APEX
},
major x tick style = {opacity=0},
minor x tick num = 1,
minor tick length=2ex,
every node near coord/.append style={
anchor=west,
rotate=90
},
legend entries={LUT-s, Registers, RAM},
legend columns=3,
legend style={draw=none,nodes={inner sep=3pt},anchor=north west, at={(0.02,0.98)}},
]
\addplot[draw=black,postaction={
pattern=dots
}, fill=blue!20, nodes near coords] table[x index=0,y index=1] \dataset;
\addplot[draw=black, postaction={
pattern=north east lines
}, fill=blue!40, nodes near coords] table[x index=0,y index=2] \dataset;
\end{axis}
\end{tikzpicture}
\vspace{-1em}
\caption{Comparison with passive RoTs: Hardware overhead}
\vspace{-1mm}
\label{fig:comparison}
\end{figure}
\section{Related Work}\label{sec:rw}
Aside from closely related work in~\cite{proactive1} and~\cite{proactive2} (already discussed in 
Section~\ref{sec:app_authWDT}), several efforts yielded \emph{passive} RoT designs for 
resource-constrained low-end devices, along with formal specifications, formal verification and provable security.

Low-end RoT-s fall into three general categories: software-based, 
hardware-based, or hybrid. Establishment of software-based RoT-s~\cite{KeJa03, SPD+04, SLS+05, SLP08, GGR09, LMP11,gligor} 
relies on strong assumptions about precise timing and constant communication delays, which can be unrealistic in the IoT ecosystem.
However, software-based RoTs are the only viable choice for legacy devices that have no security-relevant hardware support.
Hardware-based methods~\cite{PFM+04, TPM, KKW+12, SWP08, MPA08, MQY10,sancus} rely on security 
provided by dedicated hardware components (e.g., TPM~\cite{TPM} or ARM TrustZone~\cite{trustzone}). However, 
the cost of such hardware is normally prohibitive for lower-end IoT devices. Hybrid RoTs~\cite{smart,apex,vrased,tytan,trustlite} 
aim to achieve security equivalent to hardware-based mechanisms, yet with lower hardware cost. 
They leverage minimal hardware support while relying on software to reduce the complexity of additional hardware.

In terms of functionality, such embedded RoTs are passive. Upon receiving a request from an external trusted {\it Verifier}, 
they can generate unforgeable proofs for the state of the MCU or that certain actions were performed by the MCU.
Security services implemented by passive RoTs include: (1) memory integrity verification, i.e., remote attestation 
\cite{smart,sancus,vrased,simple,tytan,trustlite}; (2) verification of runtime properties, including control-flow and data-flow 
attestation~\cite{MPA08,apex,litehax,cflat,lofat,atrium,oat,tinycfa,geden2019hardware}; as well as (3) 
proofs of remote software updates, memory erasure, and system-wide resets~\cite{pure,verify_and_revive,asokan2018assured}.
As discussed in Section~\ref{sec:intro} and demonstrated in Section~\ref{sec:applications}, several application 
domains and use-cases could greatly benefit from more active RoT-s. Therefore, the key motivation for \acro 
is to not only provide proofs that actions have been performed (if indeed they were), but also to assure that 
these actions will necessarily occur.

Formalization and formal verification of RoTs for MCU-s is a topic that has recently attracted lots of attention due 
to the benefits discussed in Sections~\ref{sec:intro} and~\ref{sec:verif_background}. VRASED~\cite{vrased} 
implemented the first formally verified hybrid remote attestation scheme. APEX~\cite{apex} builds atop 
VRASED to implement and formally verify an architecture that enables proofs of remote execution of 
attested software. PURE~\cite{pure} implements provably secure services for software updates, 
memory erasure, and system-wide resets atop VRASED's RoT. Another recent result~\cite{busi2020provably} 
formalized, and proved security of, a hardware-assisted mechanism to prevent leakage of secrets through 
time-based side-channel that can be abused by malware in control of the MCU interrupts. Inline with 
aforementioned work, \acro also formalizes its assumptions along with its goals and implements the 
first formally verified active RoT design.
\section{Conclusions}\label{sec:conclusion}
This paper motivated and illustrated the design of \acro: an active RoT targeting low-end MCU-s used as platforms for 
embedded/IoT/CPS devices that perform safety-critical sensing and actuation tasks. We believe that \acro is the first 
clean-slate design of a active RoT and the first one applicable to lowest-end MCU-s, which cannot host 
more sophisticated security components, such as ARM TrustZone, Intel SGX or TPM-s.
We believe that this work is also the first formal treatment of the matter and the first active RoT to support
a wide range of RoT \trigger-s.

\small
\bibliographystyle{plain}

\bibliography{references}

\end{document}